\documentclass[a4paper,english,numberwithinsect]{lipics-v2021}
\pdfoutput=1   
\hideLIPIcs
\nolinenumbers
\usepackage[commandnameprefix=ifneeded]{changes}

\title{Improved Bounds for Discrete Voronoi Games}
\author{Mark de Berg}{Department of Mathematics and Computer Science, TU Eindhoven, the Netherlands}{M.T.d.Berg@tue.nl}{https://orcid.org/0000-0001-5770-3784}{MdB is supported by the  Dutch Research Council (NWO) through Gravitation-grant NETWORKS-024.002.003.}
\author{Geert {van Wordragen}}{Department of Computer Science, Aalto University, Espoo, Finland}{Geert.vanWordragen@aalto.fi}{https://orcid.org/0000-0002-2650-638X}{}
\authorrunning{M.~de Berg, and G. van Wordragen} 
\titlerunning{Improved Bounds for Discrete Voronoi Games}
\Copyright{Mark de Berg and Geert van Wordragen}
\relatedversion{A preliminary version of this work appeared in WADS 2023~\cite{DBLP:conf/wads/BergW23}} 

\usepackage[noend]{algpseudocode}
\usepackage{algorithm}
\usepackage{xspace}
\usepackage{amsmath, amssymb, amsfonts}
\usepackage{hyperref, cleveref, thm-restate}
\usepackage{xcolor}
\usepackage{graphicx}

\graphicspath{{./figures/}}

\renewcommand{\epsilon}{\varepsilon}
\newcommand{\etal}{\emph{et~al.}\xspace}
\newcommand{\bigO}[1]{\mathcal{O}\left(#1\right)}

\newcommand{\epsr}[1]{\epsilon_{#1}^{\square}}

\DeclareMathOperator{\dist}{dist}
\DeclareMathOperator{\myroot}{root}
\newcommand{\size}{\mathit{size}}
\newcommand{\plus}{\mathit{plus}}

\newcommand{\B}{\mathcal{B}}
\newcommand{\R}{\mathcal{R}}
\newcommand{\T}{\mathcal{T}}

\newcommand{\eps}{\varepsilon}
\renewcommand{\epsilon}{\eps}

\newcommand{\REAL}{\ensuremath{\mathbb{R}}}
\newcommand{\Reals}{\REAL}

\renewcommand{\leq}{\leqslant}
\renewcommand{\geq}{\geqslant}

\newcommand{\Vor}{\mathrm{Vor}}
\newcommand{\Pone}{$\mathcal{P}$\xspace}
\newcommand{\Ptwo}{$\mathcal{Q}$\xspace}

\newcommand{\ceil}[1]{\left\lceil #1 \right\rceil}
\newcommand{\NE}{\mbox{{\sc ne}}}
\newcommand{\SE}{\mbox{{\sc se}}}
\newcommand{\NW}{\mbox{{\sc nw}}}
\newcommand{\SW}{\mbox{{\sc sw}}}
\newcommand{\pa}{\mathrm{pa}}   

\newcommand{\Bl}{B_{\mathrm{left}}}
\newcommand{\Bu}{B_{\mathrm{up}}}
\newcommand{\Br}{B_{\mathrm{right}}}
\newcommand{\Bd}{B_{\mathrm{down}}}
\newcommand{\Rl}{\rho_{\mathrm{left}}}
\newcommand{\Ru}{\rho_{\mathrm{up}}}
\newcommand{\Rr}{\rho_{\mathrm{right}}}
\newcommand{\Rd}{\rho_{\mathrm{down}}}
\newcommand{\Vl}{V_{\mathrm{left}}}
\newcommand{\Vu}{V_{\mathrm{up}}}
\newcommand{\Vr}{V_{\mathrm{right}}}
\newcommand{\Vd}{V_{\mathrm{down}}}
\newcommand{\ba}{($\B$.1)\xspace}
\newcommand{\bb}{($\B$.2)\xspace}

\newcommand{\vl}{v_{\mathrm{left}}}
\newcommand{\vu}{v_{\mathrm{up}}}

\newcommand{\Bnew}{\B_{\mathrm{new}}}

\begin{document}

\maketitle

\begin{abstract}
In the planar one-round discrete Voronoi game, two players \Pone and \Ptwo compete over a set $V$ of $n$~voters 
represented by points in~$\Reals^2$. First, \Pone places a set $P$ of $k$ points, then
\Ptwo places a set $Q$ of~$\ell$ points, and then each voter~$v\in V$ is won by the player
who has placed a point closest to~$v$. 
It is well known that if $k=\ell=1$, then \Pone can always win $n/3$ voters and that this
is worst-case optimal. We study the setting where $k>1$ and $\ell=1$. We present lower 
bounds on the number of voters that \Pone can always win, which improve the existing
bounds for all~$k\geq 4$. As a by-product, we obtain improved bounds on small $\eps$-nets for convex ranges.
These results are for the $L_2$~metric. We also obtain lower bounds
on the number of voters that \Pone can always win when distances are measured in the $L_1$~metric.
\end{abstract}

\section{Introduction}\label{ch:introduction}

\subparagraph{Background and motivation.}
In the \emph{discrete Voronoi game}, two players compete over a set~$V$ of $n$ voters 
in~$\Reals^d$. First, player~\Pone places a set $P$ of $k$ points, then
player~\Ptwo places a set~$Q$ of $\ell$ points disjoint from the points 
in~$P$, and then each voter~$v\in V$ is won by the player who has placed 
a point closest to~$v$. In other words, each player wins the voters 
located in its Voronoi cells in the Voronoi diagram~$\Vor(P\cup Q)$. 
In case of ties, that is, when a voter~$v$ lies on the boundary between a Voronoi cell 
owned by~\Pone and a Voronoi cell owned by~\Ptwo, then $v$ is won by player~\Pone. 
Note that \Pone first places all their $k$~points and then \Ptwo places their~$\ell$ points---hence,
this is a \emph{one-round Voronoi game}---and
that $k$ and $\ell$ need not be equal. The one-round discrete Voronoi game was 
introduced by Banik~\etal~\cite{10.1007/978-3-642-22685-4_19}.

There is also a version of the Voronoi game where the players compete over a 
continuous region~\cite{10.1007/3-540-44679-6_26,10.1007/978-3-030-68211-8_9,10.1007/s00454-003-2951-4}.
For this version a multiple-round variant, where $k=\ell$ and the players
 place points alternatingly, has been studied as well.
We will confine our discussion to the discrete one-round game.
\medskip

The discrete one-round Voronoi game for $k=\ell=1$ is closely related to 
the concept of plurality points in spatial voting theory~\cite{10.2307/3689280}. In this
theory, there is a $d$-dimensional policy space,
and voters are modelled as points indicating their preferred policies.
A \emph{plurality point} is then a proposed policy that would win at
least $\ceil{n/2}$ voters against any competing policy. Phrased in terms of
Voronoi games, this means that \Pone can place a single point 
that wins at least $\ceil{n/2}$ voters against any single point 
placed by~\Ptwo. The discrete
Voronoi game with $k>1$ and $\ell=1$ can be thought of as an election
where a coalition of~$k$ parties is colluding against a single other party.
%

Another way to interpret Voronoi games is as a \emph{competitive facility-location problem},
where two companies want to place facilities so as to attract as many customers as possible, 
where each customer will visit the nearest facility. Competitive facility 
location has not only been studied in a (discrete and continuous) spatial setting, but
also in a graph-theoretic setting; see e.g.~\cite{BANDYAPADHYAY2015270,10.1007/978-3-540-75520-3_4,4100138}.

\subparagraph{Previous work.}
The one-round discrete Voronoi game leads to interesting algorithmic as well as
combinatorial problems.
\medskip

The algorithmic problem is to compute an optimal set of locations for the players. 
More precisely, for player~\Pone the goal is to compute, given a set~$V$ of $n$ voters, 
a set~$P$ of $k$ points that wins a maximum number of voters under the assumption 
that player~\Ptwo responds optimally.
For player~\Ptwo the goal is to compute, given a voter set~$V$ and a set $P$
of points placed by~\Pone, a set~$Q$ of $\ell$~points that wins as many voters
from $V$ as possible. These problems were studied in~$\Reals^1$ by 
Banik~\etal~\cite{10.1007/978-3-642-22685-4_19} for the case $k=\ell$. 
They showed that an optimal set for~\Pone can be computed in $O(n^{k-\lambda_k})$ time,
for some $0<\lambda_k<1$, and that an optimal set for~\Ptwo can be computed in $O(n)$ if the voters are given in
sorted order. The former result was improved by De~Berg~\etal~\cite{deberg2019oneround},
who presented an algorithm with $O(k^4 n)$ running time. They also showed that
in $\Reals^2$ the problem for~\Pone is $\Sigma_2^P$-hard. 
The problem for~\Pone in the special case $k=\ell=1$, is equivalent 
to finding the so-called Tukey median of~$V$. This can be done in $O(n^{d-1} + n \log n)$ time,
as shown by Chan~\cite{10.5555/982792.982853}.
\medskip

The combinatorial problem is to prove worst-case bounds on the number of voters that player~\Pone can 
win, assuming player~\Ptwo responds optimally.
Tight bounds are only known for~$k=\ell=1$, where Chawla~\etal\cite{CHAWLA2006499} 
showed the following: for any set $V$ of $n$ voters in~$\Reals^d$, player~\Pone can win 
at least $\ceil{n/(d+1)}$ voters and at most~$\ceil{n/2}$ voters, and these bounds are tight.
Situations where \Pone can win $\ceil{n/2}$ voters are particularly interesting, as these
correspond to the existence of a plurality point in voting theory. The bounds just mentioned
imply that a plurality point does not always exist. In fact, a plurality point only
exists for certain very symmetric point sets, as shown by Wu~\etal~\cite{DBLP:conf/isaac/WuLWC13}. 
De~Berg~\etal\cite{10.1145/3186990} showed how to test in $O(n \log n)$ time if a voter set admits a plurality point.
\medskip

The combinatorial problem for $k>1$ and $\ell=1$ was studied by Banik~\etal~\cite{BANIK201641}.
They showed that player~\Pone will never be able to win more than~$\left(1-\frac1{2k}\right)n$ voters,
because player~\Ptwo can always win at least half of the voters of the most crowded
Voronoi cell in~$\Vor(P)$. Banik~\etal~\cite{BANIK201641} present two methods to derive
lower bounds on the number of voters that~\Pone can always win. Below we discuss their
results in~$\Reals^2$, but we note that they generalize their methods to~$\Reals^3$.

The first method uses
a (weak) $\eps$-net for convex ranges on the voter set~$V$, that is, a point set~$N$ such that any
convex range~$R$ containing at least~$\eps n$ voters, will also contain a point from~$N$. 
Now, if $|Q|=1$ then the voters won by~\Ptwo lie in a single Voronoi cell in $\Vor(P\cup Q)$.
Since Voronoi cells are convex, this means that if we set~$P:= N$ then \Pone wins
at least~$(1-\eps)n$ voters. Banik~\etal use the $\eps$-net construction for convex ranges
by Mustafa and Ray~\cite{Mustafa2009}. There is no closed-form expression for the size of their $\eps$-net, 
but the method can give a $(4/7)$-net of size~2, for instance, and an $(8/15)$-net of size~3. 
The smallest size for which they obtain an $\eps$-net for some $\eps\leq 1/2$, which corresponds to
\Pone winning at least half the voters, is~$k=5$. Banik~\etal show that 
the $\epsilon$-net of Mustafa and Ray can be constructed in $O(kn \log^4 n)$ time.
The second method of Banik~\etal uses an $\eps$-net for disks, instead of convex sets.
This is possible because one can show that a point $q \in Q$ that wins $\alpha$ voters,
must have a disk around it that covers at least $\left\lfloor \alpha/6 \right\rfloor$ voters 
without containing a point from~$P$. Banik~\etal then present a
$(7/k)$-net for disks of size $k$, which can be constructed in $O(n^2)$ time. 
This gives a method that ensures \Pone wins at least $\left(1 - \frac{42}{k}\right)n$ voters,
which is better than the first method when $k \geq 137$.
%




\subparagraph{Our results.}
\begin{table}[t]
    \centering
    \begin{tabular}{r|c|c|c|c|c|c|l}
         metric & $k=1$ & $k=2$ & $k=3$ & $k=4$ & $k=5$ & arbitrary $k$ & reference \\ \hline
         $L_2$ & $1/3$ & $3/7$ & $7/15$ & $15/31$ & $21/41$ & $1-\frac{42}{k}$ & Banik \etal\cite{BANIK201641} \\[0.5mm]
               &       &       &         & $1/2$ & $11/21$ & $1-\frac{20\frac58 }{k}$ & this paper \\[0.5mm]
       \hline 
       $L_1$   & $1/2$ & $3/5$  & $2/3$ & $5/7$ & $3/4$ &  $1-\frac{6\frac{6}{7}}{k}$  & this paper \\[0.5mm]
    \end{tabular}
    \vspace*{2mm}
    \caption{Lower bounds on the fraction of voters that \Pone can win on any voter set in~$\Reals^2$,
              when \Pone has $k$ points and \Ptwo has a single point.
              }
    \label{tab:lowk_allbounds}
\end{table}
We study the combinatorial question---how many voters can 
player~\Pone win from any voter set $V$ of size $n$, under optimal play from~\Ptwo---in the planar setting, for $k>1$ and $\ell=1$.
We obtain the following results, where we assume that $V$ is in general position---no three voters are collinear---and
that $n$ is even.
 
In Section~\ref{sec:lowk} we present an improvement
over the $\eps$-net bounds by
      Mustafa and Ray~\cite{Mustafa2009} for convex ranges. This
      improves the results of Banik~\etal~\cite{BANIK201641}
      on the fraction of voters that \Pone can win when $k\geq 4$ and $k$
      is relatively small. We do not have a closed-form expression for the
      size of our $\eps$-net as function of~$\eps$. Theorem~\ref{thm:l2lowk}
      gives a recurrence on these sizes, and Table~\ref{tab:lowk_allbounds} shows 
      how our bounds compare to those of Banik~\etal for~$k=4,5$
      (which follow from the bounds of Mustafa and Ray~\cite{Mustafa2009}). It is particularly interesting that
      our bounds improve the smallest $k$ for which \Pone can win at least half the voters,
      from $k=5$ to $k=4$.

In Section~\ref{sec:highk} we present a new strategy for player~\Pone. Unlike the strategies
      by Banik~\etal, it is not based on $\eps$-nets. Instead, it uses a quadtree-based approach.
      By combining this approach with several other ideas and using our
      $\eps$-net method as a subroutine, we are able to show that there is
      a set $P$ of $k$ points that guarantees that \Pone wins at least $\left(1-\frac{20\frac58}{k}\right)n - 6$ voters,
      which significantly improves the $\left(1-\frac{42}{k}\right)n$ bound of Banik~\etal
\begin{figure}[t]
    \centering
    \includegraphics[width=\textwidth]{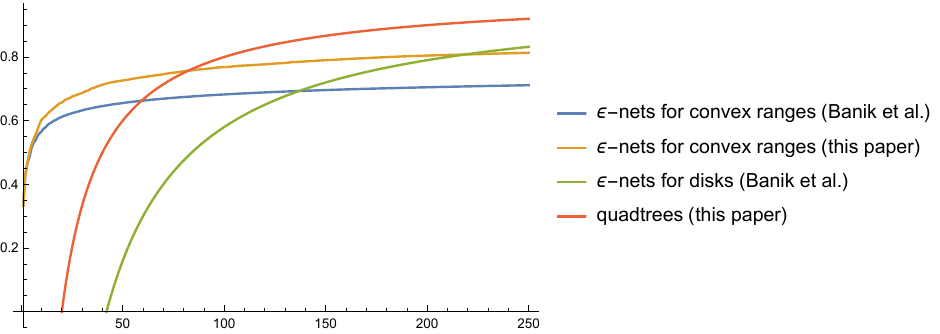}
    \caption{Lower bounds on the fraction of voters that \Pone can win 
             as a function of $k$ (the number of points of~\Pone) when \Ptwo has a single point, 
             for the $L_2$-metric. The red and green graphs do not intersect, so for large~$k$
             the quadtree method gives the best solution.}
    \label{fig:comparisonplot}
\end{figure}
Fig.~\ref{fig:comparisonplot} show the bounds obtained by the various methods in a graphical way.
\medskip

We also study the discrete one-round Voronoi game
in the $L_1$-metric, for $k>1$ and $\ell=1$. When $k=1$, player~\Pone can 
win at least half the voters 
by placing a point on a multi-dimensional median, that is, a point whose $x$- and $y$-coordinate
are medians among the $x$- and $y$-coordinates of the voter set~$V$~\cite{10.1145/3186990}.
The case $k>1$ and $\ell=1$ has, as far as we know, not been studied so far.
We first observe that for the $L_1$-metric, an $\epsilon$-net for axis-parallel rectangles
can be used to obtain a good set of points for player~\Pone. Using known results~\cite{Aronov2009}
this implies the results for $2\leq k\leq 5$ in Table~\ref{tab:lowk_allbounds}.
We also show that the quadtree-based approach
gives a lower bound of $\left( 1-\frac{6\frac{6}{7}}{k}\right) n$ on the fraction of voters that \Pone can win
in the $L_1$-metric.


\section{Better \texorpdfstring{$\varepsilon$}{ε}-nets for convex ranges}
\label{sec:lowk}
Below we present a new method to construct an $\eps$-net for convex ranges 
in the plane, which improves the results of Mustafa and Ray~\cite{Mustafa2009}. 
As mentioned in the introduction, this implies improved bounds on the number of voters
\Pone can win with $k$ points when \Ptwo has a single point, for relatively small values of~$k$.
Key to our $\eps$-net construction is the following partition result,
given for measures by Ceder~\cite{Ceder64}
(which also follows by generalising the proof of Bukh~\cite{Bukh2006} for the case $\alpha=\beta=\gamma$):
\begin{theorem}[Theorem 3 of \cite{Ceder64}] \label{thm:sixsplitcont}
    Let $\mu$ be a finite, absolutely continuous measure on~$\Reals^2$. 
    For any given $\alpha,\beta,\gamma \geq 0$ such that $2\alpha+2\beta+2\gamma=\mu(\Reals^2)$,
    we can find a set of three concurrent lines that partitions the plane into six wedges with measure $\alpha, \beta, \gamma, \alpha, \beta, \gamma$, in clockwise order.
\end{theorem}

To get such partitions for point sets, we need to introduce \emph{weights}.
In particular, assume we have a point set $V \subset \Reals^2$ with weights given
by $\omega : V \to \Reals_{\geq 0}$.
For convenience, let $\omega(S) := \sum_{v \in S} \omega(v)$ for any $S \subseteq V$ 
and assume $\omega(V) = 1$.
Consider the six closed wedges $W_1, \dots, W_6$ defined by a set of three concurrent lines.
For each wedge $W_j$, we define a weight function $\omega_j: V\rightarrow \Reals_{\geq 0}$ 
where for any point $v \in V$
we have $\omega(v) = \sum_{j=1}^6 \omega_j(v)$ and $\omega_j(v) > 0$ only when $v \in W_j$. 
Thus, points that lie in the interior of a wedge assign all of their weight to that wedge,
while the weight of a point that lies on the boundary between two wedges (or at the
common intersection of all wedges) can be distributed as needed.
We define $\omega(W_j) := \sum_{v\in V} \omega_j(v)$ to be the
total weight assigned to the wedge~$W_j$.
Theorem~\ref{thm:sixsplitcont} generalises to this discrete weighted
setting using a standard limit argument
(see e.g.\ Theorem~3.1.2 of \cite{BorsukUlam} or Appendix~A of \cite{Aronov0RT024}).
For completeness, we give the proof below.

\begin{lemma} \label{lem:sixsplit}
    Let $V$ be a set of points in $\Reals^2$ with weights given by $\omega$.
    For any given $\alpha,\beta,\gamma \geq 0$ such that $2\alpha+2\beta+2\gamma = 1$,
    we can find a set of three concurrent lines that partitions the plane into six wedges
    $W_1,\ldots,W_6$, together with a weight function $\omega_j(v)$ for each point $v\in V$
    and wedge $W_j$,
    such that the weights of the wedges are
    $\alpha, \beta, \gamma, \alpha, \beta, \gamma$, in clockwise order.
\end{lemma}
\begin{proof}
    For any $i \in \mathbb N$ and $v \in V$,
    let $D_i(v)$ be the disk of area $\omega(v) / 2^i$ around $v$.
    From this we can define a finite, absolutely continuous measure $\mu_i$
    by taking the Lebesgue measure restricted to $\bigcup_{v \in V} D_i(v)$,
    that is, $\mu_i(S) := \mathrm{area}\left(S \cap \bigcup_{v \in V} D_i(v) \right)$.
    Now consider the set $L_i$ of lines given by Theorem~\ref{thm:sixsplitcont} for $\mu_i$,
    as well as the corresponding wedges $W_1^i, \dots, W_6^i$.
    For each point $v \in V$, let $\omega_j^i(v) := 2^i \cdot \mu_i(D_i(v) \cap W_j^i)$.
    We can assume that the common intersection $p_i$ of the lines from $L_i$ 
    always lies in the convex hull of $\bigcup_{v \in V} D_0(v)$,
    which implies that $L_i$ comes from a closed and bounded set.
    (Indeed, the set can be described by the point~$p_i$ plus the three angles specifying 
    the orientations of the lines.)
    Furthermore, $\omega_j^i(v) \in [0,1]$ for any $i$, $j$ and $v$.
    Now, by the Bolzano-Weierstrass theorem, there is an infinite set $I \subseteq \mathbb N$
    where the sequences $(\omega_j^i(v))_{i \in I}$ and $(L_i)_{i \in I}$ all converge.
    We will show that their limits form a valid partition with the desired properties.

    Let $i \in \mathbb N$ be large enough to ensure that the disks $D_i(v)$ are disjoint.
    Then by construction, $L_i$ is collinear and $\sum_{j=1}^6 \omega_j^i(v) = \omega(v)$
    for any point $v \in V$.
    The weight of wedge $W_j^i$ is $\omega(W^i_j) = 2^i \cdot \mu_i(W_j^i)$,
    which will come from $\{\alpha,\beta,\gamma\}$ as prescribed.
    We can only have $\omega_j^i(v) > 0$ if $D_i(v)$ intersects $W_j^i$.
    These properties all remain true when we take the limit.
\end{proof}
We also need the following observation.
\begin{observation}\label{obs:4wedges}
    Let $L$ be a set of three lines intersecting in a common point~$p^*$, 
    and consider the six closed wedges defined by $L$.
    Any convex set $S$ not containing $p^*$ intersects at most four wedges,
    and the wedges intersected by $S$ are consecutive in the clockwise order.
\end{observation}
We now have all the tools to prove our new bounds on $\eps$-nets for convex ranges.
First, let us define what we mean by an $\eps$-net of a point set $V$ with weights $\omega$, for the family of all convex ranges:
this is a point set $N$ such that any convex range $S$
that does not contain a point from $N$ has $\omega(V \cap S) \leq \eps$.
\begin{restatable}{theorem}{Ltwolowk}\label{thm:l2lowk}
    Let $\eps_k$ be the smallest value such that any weighted point set in $\Reals^2$
    admits a $\eps_k$-net of size~$k$ for convex ranges.
    Then for any $r_1,r_2,s \in \mathbb N_0$,
    \[
        \eps_{1 + r_1 + 2r_2 + 3s} \leq \frac12 \left(
        \frac1{\eps_{r_1}} +
        \frac2{\eps_{r_2}} \right)^{-1} +
        \frac12 \epsilon_s.
    \]
\end{restatable}
\begin{figure}[t]
\centering
\includegraphics{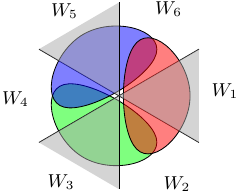}
\caption{Illustration for the proof of Theorem~\ref{thm:l2lowk}.}
\label{fig:1and2}
\end{figure}
\begin{proof}
    Let $V$ be our point set with weights $\omega$.
    We apply Lemma~\ref{lem:sixsplit} with $\alpha = \frac{\lambda}{\eps_{r_1}}$
    and $\beta = \gamma = \frac{\lambda}{\eps_{r_2}}$, where the scaling factor
    $\lambda := \frac{1}{2}\left( \frac1{\eps_{r_1}} + \frac2{\eps_{r_2}} \right)^{-1}$
    ensures $2\alpha + 2\beta + 2\gamma = 1$.
    
    Lemma~\ref{lem:sixsplit} gives us a set $L$ of three concurrent lines. We now show that there exists
    a $(\lambda + \frac12\eps_s)$-net $N$ for~$V$. 
    To this end,
    label the wedges defined by $L$ as $W_1,\dots,W_6$ in clockwise order,
    as shown in Fig.~\ref{fig:1and2}(ii).
    We can assume without loss of generality that $W_1$ and $W_4$ have weight $\alpha$, 
    that $W_2$ and $W_5$ have weight $\beta$, and that $W_3$ and $W_6$ have weight $\gamma$.
    We add the following points to our net~$N$.
    \begin{itemize}
    \item We add the common intersection of the lines in~$L$, denoted by $p^*$, to~$N$.
    \item We take an $\eps_{r_1}$-net for $V$ with weights $\omega_1$,
          an $\eps_{r_2}$-net for $V$ with weights $\omega_3$,
          and an $\eps_{r_2}$-net for $V$ with weights $\omega_5$,
          and we add the points from these nets to~$N$.
     \item For each of the three collections of three consecutive wedges---these
     are indicated in red, green, and blue in Fig.~\ref{fig:1and2}(ii)---we take an $\eps_{s}$-net
     for $V$ with the corresponding sum of weights,
     and we add the points from these nets to~$N$.
    \end{itemize}
    By construction, the size of our net~$N$ is $1+r_1+2r_2+3s$. To prove the theorem, it thus
    suffices to prove that $N$ is a $(\lambda + \frac12\eps_s)$-net.
    
    Let $S$ be any convex range, and assume $p^*\not\in S$.
    From Observation~\ref{obs:4wedges} we know that in the worst case, $S$ intersects four consecutive wedges.
    In particular $S$ cannot intersect even the boundary of the remaining two wedges,
    so it does not matter how the weight of points on those boundaries was distributed.
    Note that $N$ contains an $\eps_s$-net for $V$ with weights $\omega_s$ given by three of these wedges,
    which means that $\omega_s(V \cap S) \leq \eps_s/2$.
    For the remaining wedge, net $N$ contains an $\eps_{r_1}$-net or an $\eps_{r_2}$-net.
    Denote its weights by $\omega_r$, then
    \[
    \begin{array}{lll}
    \omega_r(V \cap S)
    & \leq & \max \left( \eps_{r_1} \alpha,\; \eps_{r_2} \beta \right) \\
    & = & \max \left(
            \eps_{r_1} \frac{\lambda}{\eps_{r_1}}, 
            \eps_{r_2} \frac{\lambda}{\eps_{r_2}} \right) \\
    & = & \lambda.
    \end{array}
    \]
    Thus, in total this is at most $\frac12\eps_s + \lambda$.
    This finishes the proof of the theorem.
\end{proof}

Note that $\eps_0=1$, since if the net is empty, a range can contain all points.
Moreover, the $\eps$-net constructions of Mustafa and Ray~\cite{Mustafa2009}
for unweighted point sets directly generalise to weighted point sets.
Namely, whenever the construction requires finding a $d$-dimensional set
containing a certain number of points we instead consider the weight,
while for everything else we use the unweighted point set.
The proof now goes through verbatim, with all resulting bounds based on weight.
This yields $\eps_1 \leq 2/3$, and $\eps_2 \leq 4/7$, and $\eps_3 \leq 8/15$. 
Using Theorem~\ref{thm:l2lowk} we can now obtain $\eps$-nets with $k\geq 4$ points,
by finding the best choice of $r_1,r_2,s$ such that $k=r_1+2r_2+3s+1$.
This gives a $(1/2)$-net using only four points, by setting $r_1,r_2=0$ and $s=1$.
Hence, player~\Pone can always place four points to win at least as many voters as player~\Ptwo, 
as opposed to the five that were proven in earlier work.
Similarly, setting $r_1=1$, $r_2=0$ and $s=1$ yields a $(10/21)$-net of five points.
\begin{corollary}\label{cor:epshalf}
   $\eps_4 \leq \frac12$ and $\eps_5 \leq \frac{10}{21}$.
\end{corollary}
\section{A quadtree-based strategy for  player~\texorpdfstring{\Pone}{P}}\label{sec:highk}

\subparagraph{The algorithm.}
First, we construct a \emph{compressed quadtree} $\T$ on the voter set~$V$.
This gives a tree structure where each node $\nu$ is associated with a square or a donut. 
We will refer to the square or donut associated to a quadtree node~$\nu$ as the \emph{cell}
of that node, and denote it by $\sigma(\nu)$. We assume that no voter in $V$ 
lies on the boundary of a cell $\sigma(\nu)$, which can be ensured by picking the square 
corresponding to $\mbox{root}(\T)$ suitably.
Donut cells in a compressed quadtree do not contain voters, and their corresponding nodes
are leaves in the compressed quadtree.\footnote{This ensures that the leaf cells partition the root cell. Note however that we do not use the donut cells, so the definition of Har-Peled~\cite{GeomAppAlg} which leaves out donut cells but is otherwise identical also suffices.}  
We denote the set of children of a node $\nu$ by $C(\nu)$.
For a square quadtree cell~$\sigma$, we denote its four quadrants 
by $\NE(\sigma)$, $\SE(\sigma)$, $\SW(\sigma)$, and $\NW(\sigma)$.

We define the \emph{size} of a square~$\sigma$, denoted by $\size(\sigma)$, to be its edge length.
Let $\dist(\sigma_1,\sigma_2)$ denote the distance between the boundaries of
two squares $\sigma_1,\sigma_2$.
The distance between two quadtree cells satisfies the following property. 
Note that the property also holds when the cells are nested.
\begin{observation} \label{obs:distance}
Let $\sigma_1$ and $\sigma_2$ be square cells corresponding to two nodes in $\T$.
If $\dist(\sigma_1,\sigma_2)>0$ then 
$\dist(\sigma_1,\sigma_2) \geq \min \left( \size(\sigma_1),\size(\sigma_2)\right)$.
\end{observation}

The idea of our algorithm to generate the $k$ points played by player~\Pone is as follows.
We pick a parameter~$m$, which depends on~$k$, and then we recursively traverse the tree~$\T$ 
to generate a set $\R$ of regions, each containing between $m+1$ and $4m$ points.
Each region $R(\nu)\in\R$ will be a quadtree cell~$\sigma(\nu)$ minus the quadtree
cells $\sigma(\mu)$ of certain nodes~$\mu$ in the subtree rooted at~$\nu$.
For each region $R\in \R$, we then generate a set of points that we put into~$P$.
The exact procedure to generate the set $\R$ of regions is described by Algorithm~\ref{alg:PlacePoints},
which is called with~$\nu = \myroot(\T)$.
\begin{algorithm}[h]
\textbf{Input:} A node $\nu$ in $\T$ and a parameter~$m$ \\
\textbf{Output:} A pair $(\R,V_\mathrm{free})$, where $\R$ is a set of regions containing at least $m+1$ and \\ 
                \hspace*{12mm} at most $4m$ voters, and $V_\mathrm{free}$ contains the voters in the subtree rooted at~$\nu$
                \\ \hspace*{12mm}  that are not yet covered by a region in $\R$.   \\[-3mm]
\begin{algorithmic}[1]
\If{$\nu$ is a leaf node}
    \State Return $(\emptyset, \{v\})$ if $\nu$ contains a voter~$v$, and return $(\emptyset, \emptyset)$ otherwise
\Else 
    \State \begin{minipage}[t]{110mm}
            Recursively call \textsc{MakeRegions}$(\mu,m)$ for all children $\mu\in C(\nu)$.
            Let $\R$ be the union of the returned sets of regions, and 
            let $V_\mathrm{free}$ be the union of the sets of returned free voters. \\[-2mm]
            \end{minipage}
    \If{$|V_\mathrm{free}| \leq m$}
        \State Return $(\R, V_\mathrm{free})$
    \Else
        \State $R(\nu) \gets \sigma(\nu) \setminus \bigcup_{R \in \R} R$  \hfill $\rhd$ {\rm Note that} $V_\mathrm{free}=R(\nu)\cap V$.
        \State Return $(\R \cup \{R(\nu)\}, \emptyset)$
    \EndIf
\EndIf
\end{algorithmic}
\caption{\emph{MakeRegions}$(\nu,m)$}\label{alg:TSP-framework}\label{alg:PlacePoints}
\end{algorithm}

We use the regions in~$\R$ to place the points for player~\Pone, as follows.
For a region $R := R(\nu)$ in $\R$, define $\sigma(R) := \sigma(\nu)$ to be the cell
of the node~$\nu$ for which~$R$ was generated. For each $R\in\R$, player~\Pone will 
place a grid of $3\times 3$ points inside~$\sigma(R)$, plus four points
outside~$\sigma(R)$, as shown in Fig.~\ref{fig:compressed-quadtree}(i).  (Some points placed for $R$
may coincide with points placed for some $R'\neq R$, but this will only
help to reduce the number of points placed by~\Pone.)

Note that each $R\in \R$ contains more than $m$ voters and the regions in $\R$ are disjoint. 
Hence, $|\R|<n/m$ and $|P|< 13n/m$. 
A compressed quadtree can be constructed in $\bigO{n\log n}$ time, and the rest of the
construction takes $\bigO{n}$ time. The following lemma summarizes the  construction.
\begin{lemma}\label{lem:alg}
The quadtree-based strategy described above places fewer than $13n/m$ points 
for player~\Pone and runs in $\bigO{n \log n}$ time.
\end{lemma}

\subparagraph{An analysis of the number of voters player~\Ptwo can win.}
To analyze the number of voters that \Ptwo can win, it will be convenient to look at
the ``child regions'' of the regions in $\R$, 
as defined next. Recall that for a region $R := R(\nu)$ in $\R$, we defined $\sigma(R) := \sigma(\nu)$.
Let $\NE(R) := R \cap \NE(\sigma(R))$ be the part of 
$R$ in the \NE-quadrant of~$\sigma(R)$. We call $\NE(R)$ a \emph{child region}
of $R$. The child regions $\SE(R)$, $\SW(R)$, and $\NW(R)$ are defined similarly;
see Fig~\ref{fig:compressed-quadtree}(ii) for an example.
\begin{figure}[t]
\centering
\includegraphics{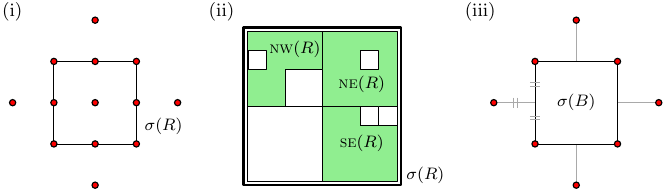}
\caption{(i) The 13 points (in red) placed in $P$ for a region $R\in\R$. 
          (ii) A region $R$ (shown in green) and its blocks (that is, its child regions). 
          The white area is covered by regions that were created before $R$.
          Since $\SW(\sigma(R))$ has already been fully covered, $\SW(R)$ does not exist.
          (iii) The eight points placed in $P$ for a type-II block $B\in\B$.
               }
\label{fig:compressed-quadtree}
\end{figure}

Let $\B$ be the set of non-empty child regions of the regions in~$\R$.
From now on, we will refer to the child regions in $\B$ as \emph{blocks}. Blocks are not 
necessarily rectangles, and they can contains holes and even be disconnected.
For a block $B\in \B$, we denote its parent region in $\R$ by $\pa(B)$,
and we let $\sigma(B)$ denote the quadtree cell corresponding to~$B$.
For instance, if $B = \NE(\pa(B))$ then $\sigma(B) = \NE(\sigma(\pa(B)))$.

Note that at the end of Algorithm~\ref{alg:PlacePoints}, the set $V_\mathrm{free}$
need not be empty. Thus the blocks in $\B$ may not cover all voters. Hence, we
add a special \emph{root block} $B_0$ to $\B$, with $\sigma(B_0) := \sigma(\myroot(\T))$ 
and which consists of the part of $\sigma(\myroot(\T))$ not covered by other blocks.
Note that we do not add any points to~$P$ for~$B_0$.

Because we will later refine our strategy for player~\Pone, it will
be convenient to analyze the number of voters that \Ptwo can win in an abstract setting.
Our analysis requires the collection $\B$ of blocks and the set $P$ of points played by \Pone
to have the following properties.
\begin{description}
\item[\ba] The blocks in $\B$ (which are still subsets of $\Reals^2$) are generated in a bottom-up manner using the compressed quadtree~$\T$.
       More precisely, there is a collection $N(\B)$ of nodes in $\T$ that is in one-to-one correspondence
       with the blocks in $\B$ such that the following holds: \\[-3mm]
       \begin{quotation}
         \noindent Let $B(\nu)$ be the block corresponding to a node $\nu\in N(\B)$. 
              Then $B(\nu) = \sigma(\nu) \setminus \bigcup_{\mu} B(\mu)$, where the union
              is taken over all nodes $\mu\in N(\B)$ that are a descendant of~$\nu$. \\[-3mm]
       \end{quotation}
         We also require that the blocks in $\B$ together cover all voters. 
\item[\bb] For each block $B\in \B$, except possibly the root block~$B_0$,
           the point set~$P$ includes the 13 points shown in Fig.~\ref{fig:compressed-quadtree}(i)
           for the cell that is the parent of $\sigma(B)$, or it includes the eight points shown
           in Fig.~\ref{fig:compressed-quadtree}(iii). In the former case we call $B$
           a \emph{type-I block}, in the latter case we call $B$ a \emph{type-II block}.
           Note that in both cases
           $P$ includes the four corners of $\sigma(B)$.
\end{description}
Observe that \ba implies that the blocks $B\in\B$ are disjoint. Moreover,
property~\bb implies the following. For a square~$\sigma$,
define $\plus(\sigma)$ to be the plus-shaped region consisting of
five equal-sized squares whose central square is~$\sigma$. 
\begin{observation}\label{obs:plus}
Let $q$ be a point played by player~\Ptwo and let $B\in\B$ be a block.
If $q$ wins a voter $v$ that lies in $\sigma(B)$, then $q\in\plus(\sigma(B))$.
Furthermore, if $q\in\sigma(B)$ then $q$ can only win voters in $\plus(\sigma(B))$.
\end{observation}
%
It is easy to see that the sets $\B$
and $P$ generated by the construction described above have properties \ba and \bb.
We proceed to analyze the number of blocks from which a point $q$ played
by~\Ptwo can win voters, assuming the set $\B$ of blocks has the properties stated above.
\medskip

We will need the following observation. It follows from \ba, which states that a block $B$
completely covers the part of $\sigma(B)$ not covered by blocks that
have been created earlier in the bottom-up process. 
\begin{observation}\label{obs:subset}
If $\sigma(B)\subset \sigma(B')$ for two blocks $B,B'\in \B$ then $B'\cap \sigma(B)=\emptyset$. 
\end{observation}
The following lemma states that the set $P$ of points played by
player~\Pone includes all vertices of the blocks in~$\B$, except possibly the
corners of the root block~$B_0$.
\begin{lemma}\label{lem:vertices-in-P}
Let $p$ be a vertex of a block $B\in \B$. Then $p\in P$, except possibly
when $p$ is a corner of $\sigma(B_0)$.
\end{lemma}
\begin{proof}
Property \ba states that the blocks in $\B$ are created in a bottom-up order.
We will prove the lemma by induction on this (partial) order.

Consider a block $B\in \B$ and let $p$ be a vertex of~$B$.
Let $s$ be a sufficiently small square centered
at $p$ and let $s_1,s_2,s_3,s_4$ be its quadrants. 
There are two cases; see Fig.~\ref{fig:vertices-in-P}.
\begin{figure}
\centering
\includegraphics{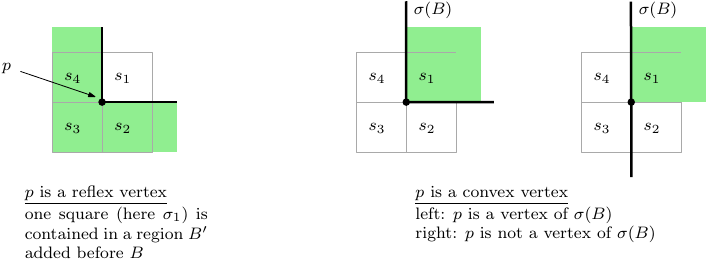}
\caption{Illustration for the proof of Lemma~\ref{lem:vertices-in-P}.}
\label{fig:vertices-in-P}
\end{figure}

If $p$ is a reflex vertex of $B$, then $B$ covers three of the four squares~$s_1,s_2,s_3,s_4$.
The remaining square must already have been covered by a region $B'$ 
created before~$B$, by Observation~\ref{obs:subset}. 
By induction, we may conclude that $p\in P$.

If $p$ is a convex vertex, then exactly one of the four squares~$s_1,s_2,s_3,s_4$,
say $s_1$, is contained in~$B$. If $p$ is a corner of $\sigma(B)$, then
$p\in P$ by property~\bb. Otherwise, at least
one square $\sigma_i\neq \sigma_1$, say $\sigma_2$, is contained in $\sigma(B)$.
We can now use the same argument as before: $p$ is a vertex of 
a region $B'$ created before~$B$,
and so $p\in P$ by induction. Note that this not only holds
when $p$ lies on an edge of $\sigma(B)$, as in Fig.~\ref{fig:vertices-in-P},
but also when $p$ lies in the interior of~$\sigma(B)$.
\end{proof}
Now consider a point $q$ played by player~\Ptwo, and assume without loss of generality
that $q\in \sigma(B_0)$. We first show that $q$ can win
voters from at most five blocks $B\in\B$; later we will improve
this to at most three blocks.
We may assume that the horizontal
and vertical lines through $q$ do not pass through a vertex of any block~$B\in\B$. 
This is without loss of generality, because an infinitesimal perturbation of $q$ ensures this
property, while such a perturbation does not change which voters are won by~$q$.
(The latter is true because voters at equal distance from $q$ and~$P$ are won by player~\Pone.)

Define $B(q)\in\B$ to be the block containing~$q$.
We start by looking more closely at which voters $q$ might win from a block $B\neq B(q)$.
Define $V(B) := V\cap B$ to be the voters lying in~$B$.
Let $\Rl$ be the axis-aligned ray emanating from $q$ and going the left, 
and define $\Ru,\Rr,\Rd$ similarly.  Let $e$ be the first edge of $B$ that is hit by~$\Rr$
and define
\[
\Vr(B) := \{ v\in V(B): \mbox{ $v$ lies in the horizontal half-strip whose left edge is $e$}  \}.
\]
Define the sets $\Vu(B)$, $\Vd(B)$, and $\Vl(B)$ similarly. 
\begin{figure}[b]
\centering
\includegraphics{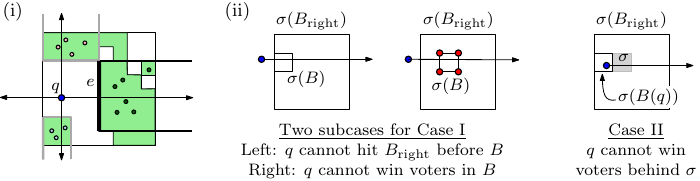}
\caption{(i) The sets of voters that $q$ might be able to win in the green block~$B$. 
        (ii) Illustration for the proof of Lemma~\ref{lem:five-child-regions}.}
\label{fig:five-child-regions}
\end{figure}
See Fig.~\ref{fig:five-child-regions}(i), where the voters from $\Vr(B)$ are shown in dark green,
the voters from $\Vu(B)$ and $\Vd(B)$ are shown in orange and blue, respectively,
and $\Vl(B)=\emptyset$. 
Because $P$ contains all vertices of~$B$ by Lemma~\ref{lem:vertices-in-P}, 
the only voters from $V(B)$ that $q$ can possibly win are the voters in 
$\Vl(B)\cup\Vu(B)\cup\Vr(B)\cup\Vd(B)$. (In fact, we could restrict these four
sets even a bit more, but this is not needed for our arguments.)

Let $\Br\neq B(q)$ be the first block in $\B$ hit by $\Rr$, and
define $\Bl,\Bu,\Bd$ similarly for the rays $\Rl,\Ru,\Rd$.
The next lemma states that there is only one block $B\neq B(q)$ for which 
$q$ might be able to win voters in $\Vr(B)$, namely~$\Br$.
Similarly, $q$ can only when voters from $\Vl(B)$ for $B=\Bl$, and so on.
\begin{lemma}\label{lem:five-child-regions}
If $q$ wins voters from $\Vr(B)$, where $B\neq B(q)$, then $B=\Br$. 
\end{lemma}
\begin{proof}
Suppose for a contradiction that $q$ wins voters from $\Vr(B)$ for some block
$B\not\in \{B(q),\Br\}$. We distinguish two cases.
\\[2mm]
\emph{Case~I: $q\not\in \sigma(\Br)$. $\mbox{}$}  \\[1mm]
        Since the corners of $\sigma(\Br)$ are in~$P$ by~\bb,
        the point~$q$ cannot win voters to the right of~$\sigma(\Br)$. 
        Hence, if $q$ wins voters from $\Vr(B)$, then $B$ must lie at least partially inside $\sigma(\Br)$.
        Now consider $\sigma(B)$. We cannot have $\sigma(\Br)\subset \sigma(B)$
        by Observation~\ref{obs:subset}. Hence, $\sigma(B)\subset \sigma(\Br)$
        and so $\Br \cap \sigma(B)=\emptyset$.
        We now have two subcases, illustrated in Fig.~\ref{fig:five-child-regions}(ii).
        \begin{itemize}
        \item If the left edge of $\sigma(B)$ is contained in the left edge of $\sigma(\Br)$,
            then $\Rr$ would hit $B$ before $\Br$, contradicting the definition of~$\Br$.
        \item On the other hand, if the left edge of $\sigma(B)$ is not contained 
            in the left edge of $\sigma(\Br)$, then $\dist(\sigma(\Br),\sigma(B))\geq \size(\sigma(B))$
            by Observation~\ref{obs:distance}. Since $P$ contains the four corners
            of $\sigma(B)$, this contradicts that $q$ wins voters from~$\Vr(B)$.
        \end{itemize}
\emph{Case~II: $q\in \sigma(\Br)$. $\mbox{}$}  \\[1mm]
        We cannot have $\sigma(\Br)\subset \sigma(B(q))$,
        otherwise $B(q)\cap \sigma(\Br) = \emptyset$ by Observation~\ref{obs:subset}, 
        which contradicts $q\in B(q)$.
        Hence, $\sigma(B(q))\subset \sigma(\Br)$ and $\Br\cap \sigma(B(q))=\emptyset$.        

        Consider the square~$\sigma$ with the same size of $\sigma(B(q))$ and 
        immediately to the right of~$\sigma(B(q)))$; see the grey square in
        Fig.~\ref{fig:five-child-regions}(ii). We must have $\sigma\subset \sigma(\Br)$,
        otherwise the right edge of $\sigma(B(q))$ would be contained in the right
        edge of $\sigma(\Br)$ and so $\Rr$ would exit $\sigma(\Br)$ before it can hit~$\Br$.
        By Observation~\ref{obs:plus}, point $q$ cannot win voters to the right of $\sigma$.
        Hence, $B\cap \sigma\neq\emptyset$. 
        Now consider the relative position of $\sigma(B)$ and $\sigma$.
        There are two subcases.
        \begin{itemize}
        \item  If $\sigma(B)\subset \sigma$, then either the distance from $q$ to $B$
            is at least~$\size(\sigma(B))$ by Observation~\ref{obs:distance},
            contradicting that $q$ wins voters from~$\Vr(B)$; or $\sigma(B)$ lies immediately
            to the right of $\sigma(B(q))$, in which case $\Rr$ cannot hit $\Br$ before $B$.
        \item Otherwise, $\sigma\subset \sigma(B)$. If $\sigma(B)\subset \sigma(\Br)$, 
            then $\Br\cap \sigma(B)=\emptyset$ by Observation~\ref{obs:subset}, contradicting 
            (since $\sigma \subset \sigma(B)$) that $\Rr$ hits $\Br$ before $B$. 
            Hence, $\sigma(\Br)\subset \sigma(B)$.
            But then $B\cap \sigma(\Br)=\emptyset$, which contradicts that $B\cap\sigma\neq\emptyset$. 
        \end{itemize}
\end{proof}
Lemma~\ref{lem:five-child-regions} implies that $q$ can only win voters from
the five blocks $B(q)$, $\Bl$, $\Br$, $\Bu$, and $\Bd$. The next lemma shows that
$q$ cannot win voters from all these blocks simultaneously.
\begin{lemma}\label{lem:three-child-regions}
Point $q$ can win voters from at most three of the blocks $B(q)$, $\Bl$, $\Br$, $\Bu$, and~$\Bd$.
\end{lemma}
\begin{proof}
First suppose that the size of $\sigma(B(q))$ is at most the size of any of the
four cells $\sigma(\Bl),\ldots,\sigma(\Bd)$ from which $q$ wins voters. By Observation~\ref{obs:subset},
this implies that all four blocks
$\Bl,\Br,\Bu,\Bd$ lie outside~$\sigma(B(q))$. Then it is easy to see 
that $q$ can win voters from at most two of the four blocks $\Bl,\Br,\Bu,\Bd$,
because all four corners of $\sigma(B(q))$ are in~$P$ by~\bb.
For instance, if $q$ lies in the \NE-quadrant of $\sigma(B(q))$, then
$q$ can only win voters from $\Br$ and $\Bu$; the other cases are symmetrical.

Now suppose that $\sigma(B(q))$ is larger than $\sigma(\Br)$, which we assume without
loss of generality to be a smallest cell from which $q$ wins voters
among the four cells $\sigma(\Bl),\ldots,\sigma(\Bd)$. 
We have two cases. 
\\[2mm]
\emph{Case~I: $q\not\in \sigma(\pa(\Br))$. $\mbox{}$}  \\[1mm]
      Note that $\sigma(\Br)$ must either be the \NW- or \SW-quadrant of $\sigma(\pa(\Br))$,
      because otherwise $q\not\in\plus(\Br)$ and $q$ cannot win voters from $\Br$
      by Observation~\ref{obs:plus}. Assume without loss of generality that $\sigma(\Br) = \NW(\sigma(\pa(\Br)))$.
      Then~$q$ must be located in the square of the same size as~$\sigma(\Br)$ and immediately 
      to its left. In fact, $q$ must lie in the right half of this square.
      We now define two blocks, $\sigma$ and $\sigma'$ that
      play a crucial role in the proof. Their definition depends on whether $\Br$
      is a type-I or a type-II block.
     \begin{figure}[b]
     \centering
     \includegraphics{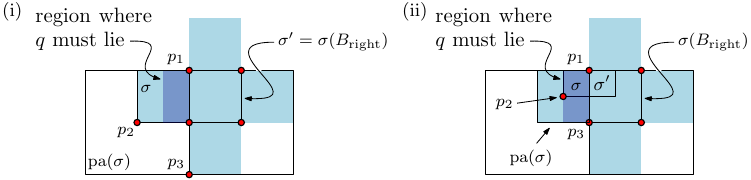}
     \caption{Two cases for the definition of $\sigma'$ and $\sigma$, 
              (i) when $\Br$ is a type-I block and (ii) when $\Br$ is a type-II block.
              In the latter case $\sigma'$ and $\sigma'$ could also lie in the bottom half of their parent
              regions, depending on where $q$ lies.}
     \label{fig:three-child-regions-sigma-prime}
     \end{figure}
      \begin{itemize}
      \item If $\Br$ is a type-I block, then we define $\sigma' := \sigma(\Br)$ and we define
            $\sigma$ to be the square of the same size as~$\sigma'$ and immediately to its left.
            Note that $q\in \sigma$, since $q$ wins voters from $\sigma(\Br)$.
            See Fig.~\ref{fig:three-child-regions-sigma-prime}(i).
      \item  If $\Br$ is a type-II block, then we define $\sigma' := \NW(\sigma(\Br))$ 
             or $\sigma' = \SW(\sigma(\Br))$ and we define $\sigma$ to be the square of the same 
             size as~$\sigma'$ and immediately to its left. Whether we choose
             $\sigma' := \NW(\sigma(\Br))$ or $\sigma' = \SW(\sigma(\Br))$ depends on the
             position of~$q$: the choice is made such that the square $\sigma$ to the
             left of $\sigma'$ contains~$q$. See Fig.~\ref{fig:three-child-regions-sigma-prime}(ii)
             for an example. Since we will not use $\pa(\pa(\sigma'))$ 
             in the proof, these two choices are symmetric as far as the proof is concerned---we
             only need to swap the up- and down-direction.
      \end{itemize}
      We now continue with the proof of Case~I. All statements referring to $\sigma$ and $\sigma'$
      will hold for both definitions just given.
      
      Observe that $\sigma(\Bd)\neq \sigma$, since otherwise $\sigma(\Bd)\subset \sigma(B(q))$,
      contradicting by Observation~\ref{obs:subset} that~$q\in B(q)$. We will now consider 
      three subcases. In each subcase we argue that either we are done---we will have shown
      that $q$ wins voters from at most one of the blocks $\Bd$, $\Bl$, and $\Bu$---or
      $q$ cannot win voters from $\Bd$. After discussing the three subcases, we then
      continue the proof under the assumption that $q$ does not win voters from $\Bd$.
      \begin{itemize}
      \item \emph{Subcase~(i): $\sigma(\Bd)=\SE(\pa(\sigma))$ and $\Bd$ is a type-I block.} \\
            In the case all four corners of $\sigma$ are in $P$. If $q\in\NE(\sigma)$ 
            then $q$ can only win voters from $\Bu$ and if $q\in\SE(\sigma)$ then $q$ can 
            only win voters from $\Bd$  (this is in addition to voters won from $B(q)$ and $\Br$),
            and so we are done.
      \item \emph{Subcase~(ii): $\sigma(\Bd)=\SE(\pa(\sigma))$ and $\Bd$ is a type-II block.} \\
            If $q\in \SE(\sigma)$ then $q$ cannot win voters from $\Bu$ or $\Bl$, and so we are
            done. Otherwise $q$ cannot win voters from $\Bd$, as claimed.
      \item \emph{Subcase~(iii): $\sigma(\Bd)\neq \SE(\pa(\sigma))$.} \\
             If $\Bd\cap \pa(\sigma)\neq\emptyset$ then both $\Bd$ and $B(q)$ intersect~$\pa(\sigma)$,
             and both $\sigma(\Bd)$ and $\sigma(B(q))$ contain $\pa(\sigma)$. But this is
             impossible due to Observation~\ref{obs:subset}. Hence, $\Bd$ must
             lie below $\pa(\sigma)$.
             We claim that then $q$ cannot win voters from $\Bd$.
             The closest $q$ can be to $\Bd$ is when it lies on the bottom line segment of $\sigma$. 
             Hence, any voter in $\Bd$ won by~$q$ must be closer to that segment than 
             to $p_3$, and also than to~$p_2$. (See Fig.~\ref{fig:three-child-regions-sigma-prime}
             for the locations of $p_2$ and $p_3$.) But this is clearly impossible.
             Hence, $q$ cannot win voters from $\Bd$.
      \end{itemize}
      Thus, in the remainder of the proof for Case~I we can assume 
      that $q$ does not win voters in $\Bd$. Hence, it suffices
      to show that $q$ cannot win voters from $\Bu$ and $\Bl$
      simultaneously. To this end, we assume $q$ wins a voter $\vu$ from $\Bu$ and
      a voter $\vl$ from $\Bl$ and then derive a contradiction.
      
      Let $e_{\mathrm{up}}$ be the first edge of $\Bu$ hit by $\Ru$ and let
      $e_{\mathrm{left}}$ be defined analogously; see Fig.~\ref{fig:eup-and-eleft}.
      Note that $e_{\mathrm{up}}$ and $e_{\mathrm{left}}$ must lie outside $\pa(\sigma)$, otherwise we
      obtain a contradiction with Observation~\ref{obs:subset}. 
     \begin{figure}[b]
     \centering
     \includegraphics{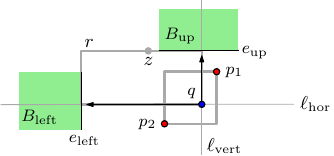}
     \caption{Definition of $e_{\mathrm{left}}$ and $e_{\mathrm{up}}$, and $r$ and $z$.}
     \label{fig:eup-and-eleft}
     \end{figure}

      Let $\ell_{\mathrm{hor}}(q)$ be the horizontal line through~$q$. 
      \begin{claim}
       $\vl$ must lie above $\ell_{\mathrm{hor}}(q)$.
      \end{claim}
      \emph{Proof of Claim.} 
      We need to show that the perpendicular bisector of $q$ and $p_2$ 
      will always intersect the left edge of $\pa(\sigma)$ above~$\ell_{\mathrm{hor}}(q)$.
      For the situation in Fig.~\ref{fig:three-child-regions-sigma-prime}(i) this is relatively
      easy to see, since $q$ lies relatively far to the left compared to~$p_2$.
      For  situation in Fig.~\ref{fig:three-child-regions-sigma-prime}(ii),
      it follows from the following argument. To win voters in $\Br$, the point~$q$
      must lie inside the circle $C$ through $p_1,p_2,p_3$. Now, suppose $q$ actually lies on~$C$
      and let $\alpha := \angle q zp_2$, where $z$ is the center of $C$. Thus
      the bisector of $q$ and $p_3$ has slope $-\tan(\alpha/2)$. The result then
      follows from the fact that $2 \tan (\alpha/2)-\sin \alpha>0$ for $0<\alpha<\pi/2$.
      \qed
      \medskip
      
      \noindent For $\vu$ the situation is slightly different: $q$ can, in fact,  win voters to 
      the right of $\ell_{\mathrm{vert}}(q)$, the vertical line through~$q$. In that case,
      however, it cannot win $\vl$.
      \begin{claim}
       If $q$ wins a voter from $\Bu$ right of $\ell_{\mathrm{vert}}(q)$, then
       $q$ cannot win a voter from $\Bl$.
      \end{claim}
      \emph{Proof of Claim.}
          It follows from Observation~\ref{obs:distance} 
          that $e_{\mathrm{up}}$ must overlap with the top edge of~$\sigma$.
          Because the edges $e_{\mathrm{up}}$ and $e_{\mathrm{left}}$ cannot intersect, one
          of them must end when or before the two meet. 
          
          If $e_{\mathrm{left}}$ ends before meeting (the extension of) $e_{\mathrm{up}}$,
          then the top endpoint of $e_{\mathrm{up}}$, which is in $P$
          by Lemma~\ref{lem:vertices-in-P}, prevents $q$
          from winning voters from $\Bl$. 
          
          So now assume that $e_{\mathrm{up}}$ ends before meeting (the extension of) $e_{\mathrm{left}}$.
          Then the left endpoint of $e_{\mathrm{up}}$, which we denote by~$p_4$, is in~$P$.
          Without loss of generality, set $p_1 = (0,1)$, $p_2 = (-1, 0)$ and $p_4 = (p_x,1)$.
          Now, winning voters from $\Bu$ means $q$ must lie inside the circle
          $C_{\mathrm{up}}$ with center on $e_{\mathrm{up}}$ that goes through
          $p_1$ and $p_4$. Thus it has center $c = (\frac{p_x}{2}, 0)$.
          It must also lie in the circle $C$ through $p_1,p_2,p_3$
          so it can win voters from $\Br$. The line through the circle centers
          makes an angle $\alpha := \arctan \frac{p_x}{2}$ with the line $x=0$.
          The circles intersect at~$p_1$, which lies on the line $x=0$,
          so their other intersection point lies on the line $\ell$
          that makes an angle $2\alpha$ with $x=0$.
          If $q$ is to win voters from both $\Br$ and $\Bu$
          it must lie between $\ell$ and $x=0$. Next we show that
          this implies that $q$ cannot win voters from $\Bl$.
          
          We first show that if $q=\ell\cap C$, then $p_4$ prevents $q$ from winning 
          voters in $\Bl$.
          Because then $p_4$ and $q$ both lie on $C_{\mathrm{up}}$, their perpendicular
          bisector $b(p_4,q)$ is the angular bisector of $\angle q c p_4$.
          Note that $\angle q c p_4 = 2\alpha$. Indeed, the line through the circle centers
          makes a right-angled triangle together with $y=1$ and $x=0$, so the
          angle at $c$ must be $\frac{1}{2} \pi - \alpha$. 
          Hence, $\angle q c p_1 = \pi-2\alpha$, and so $\angle q c p_4 = 2\alpha$.
          Thus,  $b(p_4,q)$ makes an angle $\alpha$ with $y=1$ and so it intersects the line $x=-2$ at height
          $y= 1 - (2+\frac{p_x}{2}) \tan\alpha$ which is $1 - (2 + \tan\alpha) \tan\alpha$.
          For $0 < 2\alpha < \pi/2$ this is below $\ell_{\mathrm{hor}}(q)$
          which lies at $y = \cos 2\alpha$. 
          By the previous Claim, this means that $q$ cannot win voters from $\Bl$.
          Therefore, $q$ cannot win voters from $\Bl$.

          To finish the proof, we must argue that $q$ cannot win voters from $\Bl$ either
          when $q\neq \ell\cap C$. It clear that moving $q$ to the left helps to win voters
          in $\Bl$, so we can assume that $q\in C$. Then it is not hard to see (by following the calculations 
          above)
          that the best position for $q$ is $\ell\cap C$, for which we just showed
          that $q$ cannot win voters in $\Bl$. This finishes the proof of the claim.
      \qed 
      \medskip
      
      \noindent We can now assume $\vl$ lies above $\ell_{\mathrm{vert}}(q)$ and
      $\vu$ lies to the left of $\ell_{\mathrm{vert}}(q)$. We will show that this leads
      to a contradiction. To this end, consider the rectangle~$r$ whose bottom-right corner
      is $q$, whose top edge overlaps $e_{\mathrm{up}}$ and whose left edge
      overlaps with $e_{\mathrm{left}}$; see Fig.~\ref{fig:eup-and-eleft}. 
      Then the left edge of~$r$ contains the top endpoint of $e_{\mathrm{left}}$ 
      and/or the top edge of~$r$ contains the left endpoint of $e_{\mathrm{up}}$.
      By Lemma~\ref{lem:vertices-in-P}, we thus know that there is
      a point $p_4\in P$ lying on the left or top edge of~$r$.
      Now assume without loss of generality that the top edge of~$r$
      is at least as long as its left edge, and let $z\in e_{\mathrm{up}}$ be the point
      such that the $q_x-z_x=z_y-q_y$. Now, if $p_4$ lies on the left edge of~$r$
      or to the left of $z$ on the top edge, then $p_4$ prevents $q$ from winning~$\vl$.
      On the other hand, if $p_4$ lies to the right of $z$ on the top edge of~$r$,
      then $p_4$ prevents $q$ from winning $\vu$.
      So in both cases we have a contradiction.
      \\[2mm]
\emph{Case~II: $q\in \sigma(\pa(\Br))$. $\mbox{}$}  \\[1mm]
    Assume without loss of generality that $\sigma(\Br)$ is one of the two northern
    quadrants of~$\pa(\sigma(\Br))$. We cannot have $q\in \sigma(\Br)$,
    since together with $\size(\sigma(\Br))<\size(\sigma(B(q)))$ this contradicts
    $q\in B(q)$, by Observation~\ref{obs:subset}.
    Hence, $q\in \NW(\pa(\sigma(\Br)))$ and $\sigma(\Br)=\NE(\pa(\sigma(\Br)))$.
    
    If $\Br$ is a type-I block then all corners of $\NW(\pa(\sigma(\Br)))$ are in $P$, 
    which (as we saw earlier) implies that $q$ can win voters from at
    most three blocks. 
    If $\Br$ is a type-II block, then we can follow the proof of Case~I.
    (For type-I blocks this is not true. The reason is that 
    in the proof of the first Claim, we use that $\Bl$ does not lie immediately
    to the left of $\sigma$, which is not true for type-I blocks in Case~2.
    Note that this still is true for type-II blocks in Case~2.)
\end{proof}
By construction, each block contains at most $m<n/|\R|$ voters, where $\R$ is the set
of regions created by Algorithm~\ref{alg:PlacePoints}.
Moreover, \Pone places 13 points per region in~$\R$, and so $k \leq 13|\R|$ points in total. 
Finally, Lemma~\ref{lem:three-child-regions} states that \Ptwo can win voters from at most three blocks.
We can conclude the following.
\begin{lemma}\label{lem:3m}
Let $V$ be a set of $n$ voters in $\mathbb R^2$. For any given $k$, the quadtree-based strategy
described above can guarantee that \Pone wins at least $\left(1-\frac{39}{k}\right)n$
voters by placing at most $k$ points, against any single point placed by player~\Ptwo.
\end{lemma}

\subparagraph{A more refined strategy for player~\Pone.}
It can be shown that the analysis presented above is tight.
Hence, to get a better bound we need a better strategy.

Recall that each region $R\in\R$ contains between $m+1$ and $4m$ voters.
Currently, we use the same $13$ points for any~$R$, regardless of the 
exact number of voters it contains and how they are distributed over the child regions of $R$.
Our refined strategy takes this into account, and also incorporates the $\eps$-nets
developed in the previous section, as follows.
Let $n_R$ denote the number of voters in a region~$R\in \R$.
We consider two cases, with several subcases.
\begin{itemize}
\item \emph{Case~A: $m<n_R \leq \frac{16}{11}m$.} 
      We place eight points in total for $R$, as in Fig.~\ref{fig:compressed-quadtree}(iii).
      We also add between two and six extra points, depending on the subcase.
      \begin{itemize}
      \item If $m<n_R \leq \frac{7}{6}m$, we add two extra points, forming a $\frac{4}{7}$-net.
        \item If $\frac{7}{6}m <n_R \leq \frac{5}{4}m$, we add three extra points, forming a $\frac{8}{15}$-net.
        \item If $\frac{5}{4}m <n_R \leq \frac{4}{3}m$, we add four extra points, forming a $\frac{1}{2}$-net.
        \item If $\frac{4}{3}m <n_R \leq \frac{7}{5}m$, we add five extra points, forming a $\frac{10}{21}$-net.
        \item If $\frac{7}{5}m <n_R \leq \frac{16}{11}m$, we add six extra points, forming a $\frac{11}{24}$-net.\\[-2mm]
     \end{itemize}
     One can show that in each subcase above, player~\Ptwo wins at most
     $2m/3$
     voters from inside~$R$,
     due to the
     $\eps$-nets. For example, in the first case \Ptwo wins at most $(7m/6)\cdot (4/7)=2m/3$ voters,
     in the second subcase  \Ptwo wins at most $(5m/4)\cdot (8/15)=2m/3$ voters, etcetera.
     Furthermore, one easily verifies that in each subcase we have 
     $\frac{\mbox{\tiny number of voters in $R$}}{\mbox{\tiny number of points placed}} > m/10$. \\
\item \emph{Case~B: $\frac{16}{11}m < n_R \leq 4m$.} 
       We first place the same set of 13 points as in our original strategy.
       We add two or four extra points, depending on the subcase. \\[-2mm]
       \begin{itemize}
      \item If $\frac{16}{11}m <n_R \leq 2m$ we add two extra points, as follows.
            Consider the four child regions of~$R$. Then we add a centerpoint---in
            other words, a $\frac{2}{3}$-net of size~1---for the voters in
            the two child regions with the largest number of voters.
%
%
\item If $2m <n_R \leq 4m$ we add four extra points, namely a centerpoint for each of the four
       child regions of $R$. \\[-2mm]
       \end{itemize}
       Note that in both subcases, \Ptwo wins at most $2m/3$ voters from any child region. For the
       child regions where we placed a centerpoint, this holds because a child region contains
       at most $m$ voters by construction. For the two child regions where we did not place 
       a centerpoint in the first subcase, this holds because these child regions  
       contains at most $2m/3$ voters.  
       Furthermore, in both subcases 
       $\frac{\mbox{\tiny number of voters in $R$}}{\mbox{\tiny number of points placed}} > \frac{16}{165}m$.
\end{itemize}
\begin{theorem}\label{lem:refines}
    Let $V$ be a set of $n$ voters in $\mathbb R^2$. For any given $k$, the refined quadtree-based strategy
   can guarantee that \Pone wins at least
   $\left(1-\frac{20\frac{5}{8}}{k}\right)n$
    voters by placing at most $k$ points,  against any single point placed by player~\Ptwo.
\end{theorem}
\begin{proof}
   The proof for the original quadtree-based strategy was based on two facts: 
   First, player~\Ptwo can win voters from at most three blocks $B\in\B$;
   see Lemma~\ref{lem:three-child-regions}. Second, any block $B\in\B$
   (which was a child region of some $R\in\R$) contains at most $m$ voters.

   In the refined strategy, we use a similar argument, but for a set $\Bnew$ of blocks
   defined as follows. For the regions $R\in \R$ that fall into Case~A,
   we put $R$ itself (instead of its child regions) as a type-II block into $\Bnew$.
   For the regions $R\in\R$ that fall into Case~B, we put their child regions
   as type-I blocks into~$\Bnew$.
   By Lemma~\ref{lem:three-child-regions}, \Ptwo can win voters from at most three blocks in~$\Bnew$.
   Moreover, our refined strategy ensures that \Ptwo wins at most
   $2m/3$
   voters from
   any $B\in\Bnew$. Thus \Ptwo wins at most
   $2m$
   voters in total. 
   
   Finally, for each $R\in \R$ we have 
   $\frac{\mbox{\tiny number of voters in $R$}}{\mbox{\tiny number of points placed}} > \frac{16}{165}m$.
   Hence, $m < \frac{165}{16k}n$ and so \Ptwo wins at most 
   $\frac{165}{8k}n = \frac{20\frac{5}{8}}{k}n$
   voters.
\end{proof}

\section{Strategies for the \texorpdfstring{$L_1$}{L1}-metric}
\label{sec:L1persL1}
We now consider the setting where the distances from
a voter to the points placed by \Pone or \Ptwo are measured in the $L_1$-metric.
We can use similar techniques as for $L_2$-metric, as explained next.

\subparagraph{Using $\eps$-nets.}\label{sec:L1low}
Recall that the point $q$ played by \Ptwo wins voters from the Voronoi cell of $q$ in $\Vor(P\cup\{q\})$. 
Unfortunately, Voronoi cells under the $L_1$-metric are not convex but only star-shaped, 
so we cannot use $\epsilon$-nets in the same way as before.
Instead, we can use the following observation. 
Let $p \in \Reals^2$ be a point played by player~\Pone, and consider the four quadrants defined
by the vertical line and the horizontal line containing~$p$. Then, in the $L_1$-metric,  a point~$q$ played by \Ptwo
cannot win voters from two opposite quadrants.
In other words, for each point $p$ played by \Pone, player~\Ptwo can only win voters in 
one of the two half-planes defined by the horizontal line containing~$p$, or in 
one of the two half-planes defined by the horizontal line containing~$p$. By taking the
intersection of the relevant half-planes defined by the points $p\in P$, we see that the
voters won by \Ptwo lie in an axis-parallel rectangle that does not contain a point from~$P$.
Hence, we have the following lemma.
\begin{lemma}
Let $V$ be a set of $n$ voters in $\mathbb R^2$.
If~\Pone places points according to an $\epsilon$-net for $V$ with respect to axis-parallel rectangles, 
then~\Ptwo always wins at most $\epsilon n$ voters.
\end{lemma}
It is known~\cite{Aronov2009,dulieu2006,Rachek2020} that there is a $(2/5)$-net of size two, a $(1/3)$-net of size three, a $(2/7)$-net
of size four, and a $(1/4)$-net of size 5, leading to the bounds in Table~\ref{tab:lowk_allbounds}.
Aronov \etal\cite{Aronov2009} conjecture that there is an $\frac{2}{k+3}$-net with $k$ points for any $k\geq 1$.
If true, this would mean this technique is also effective for large~$k$.

\subparagraph{The quadtree-based strategy in the $L_1$-metric.}
The quadtree-based strategy can also be applied for the $L_1$-metric.
Here it gives better bounds than for the $L_2$-metric,
because (as we will show) \Ptwo wins voters from at most two blocks
instead of three. Hence, \Ptwo can only win  $26 \frac{n}{k}$ voters,
instead of $39 \frac{n}{k}$. As in the $L_2$ metric, we further improve this
by using different types of blocks depending on the number and 
distribution of the voters, and by using $\eps$-nets. This leads to the
following theorem.
\begin{restatable}{theorem}{LOneBoundTwo}
\label{thm:L1-quadtree}
Let $V$ be a set of $n$ voters in $\mathbb R^2$. For any given $k$, player~\Pone
can place $k$ points such that \Pone wins at least $\left(1-\frac{6\frac{6}{7}}{k}\right)n$
voters against any single point placed by player~\Ptwo,
in the $L_1$-metric.
\end{restatable}
In the remainder of this section we will prove Theorem~\ref{thm:L1-quadtree} by analyzing
the quadtree-based strategy in the $L_1$-metric. The analysis is similar
to the analysis in the Euclidean metric, but the bounds are better due to the fact
that distances are now measured in the $L_1$-metric.

We start by proving a stronger version of Observation~\ref{obs:plus}.
Recall that set $P$ of points placed by \Pone are generated based on
a collection~$\B$ of blocks with properties \ba and \bb; see also
Fig.~\ref{fig:compressed-quadtree}.

\begin{figure}[ht]
    \centering
    \includegraphics{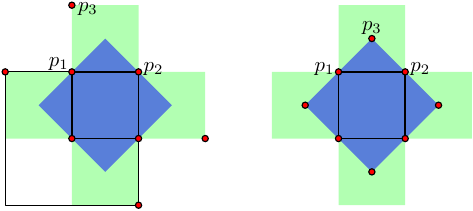}
    \caption{A type~I block (left) and a type-II block (right).
    For both, a point $q$ that wins voters from the block must lie in the blue diamond shape and can only win further voters from the green plus shape.}
    \label{fig:l1_plus}
\end{figure}

\begin{lemma}\label{lem:l1_plus}
If $q$ wins a voter that lies in $\sigma(B)$, then any other voter won by $q$ must lie in $\plus(\sigma(B))$.
\end{lemma}
\begin{proof}
If $q \in \sigma(B)$, then this holds by Observation~\ref{obs:plus}.
To handle the case $q\not\in\sigma(B)$, observe that in order to win a voter
in $\sigma(B)$, the point~$q$ must lie in the diamond shape shown in blue
in Fig.~\ref{fig:l1_plus}. Assume wlog that $q$ lies in the blue triangle above $p_1 p_2$
in the figure. Then it is easy to see that $q$ cannot win voters outside $\plus(\sigma(B))$
due to the points $p_1,p_2,p_3$ placed by \Pone.
\end{proof}

Lemma~\ref{lem:five-child-regions} still holds because it uses no properties of the $L_2$-metric other than the triangle inequality.
Combined with  Lemma~\ref{lem:l1_plus}, this lets us prove a stronger version of Lemma~\ref{lem:three-child-regions}:
\begin{lemma}\label{lem:l1_noopposite}
Point $q$ can win voters from at most two of the blocks $B(q)$, $\Bl$, $\Br$, $\Bu$, and~$\Bd$.
\end{lemma}
\begin{proof}
First suppose that the size of $\sigma(B(q))$ is at most the size of any of the
four cells $\sigma(\Bl),\ldots,\sigma(\Bd)$ from which $q$ wins voters. By Observation~\ref{obs:subset},
this implies that all four blocks $\Bl,\Br,\Bu,\Bd$ lie outside~$\sigma(B(q))$.
Indeed, if $\Bl$ does not lie outside $\sigma(B(q))$, for instance, then it must be that $\sigma(B(q)) \subset \sigma(\Bl)$, but then Observation~\ref{obs:subset} implies $\Bl \cap \sigma(B(q)) = \emptyset$.
It can be seen in Fig.~\ref{fig:l1_two_regions}(i)
that $q$ can only win voters from one of the four blocks $\Bl,\Br,\Bu,\Bd$,
because all four corners of $\sigma(B(q))$ are in~$P$ by~\bb.

Now suppose that $\sigma(B(q))$ is larger than $\sigma(\Br)$, which we assume without
loss of generality to be a smallest cell from which $q$ wins voters
among the four cells $\sigma(\Bl),\ldots,\sigma(\Bd)$. 
Let $\sigma$ be the square of the same size as $\sigma(\Br)$ placed directly left of $\sigma(\Br)$, as in Fig.~\ref{fig:l1_two_regions}(ii).
None of $\Bl,\Bu,\Bd$ can intersect $\sigma$.
To show this, assume there is a block $B$ that does intersect~$\sigma$.
Then, either $\sigma \subset \sigma(B) \subset \sigma(B(q))$ or $\sigma \subset \sigma(B(q)) \subset \sigma(B)$.
However, by Observation~\ref{obs:subset} that means that
in the first case $B(q) \cap \sigma(B) = \emptyset$, but $q \in \sigma \subset \sigma(B)$ by Observation~\ref{obs:plus} so this is incompatible with the definition of $B(q)$.
In the second case $B \cap \sigma(B(q)) = \emptyset$, which contradicts the assumption that $B$ intersects $\sigma$.
Thus, none of $\Bl,\Bu,\Bd$ can intersect $\sigma$.
Now, Lemma~\ref{lem:l1_plus} implies that $q$ cannot win voters in $\Bl$, $\Bu$ or $\Bd$.
\end{proof}

\begin{figure}[t]
    \centering
    \includegraphics{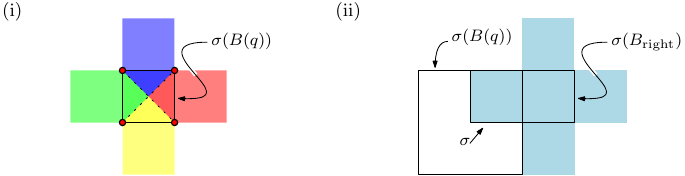}
    \caption{(i) When $q$ is in a coloured triangle of $\sigma(B(q))$, it can only win voters from the same-coloured square. (ii) Definition of $\sigma$, compared to $\sigma(B(q))$ and $\sigma(\Br)$.}
    \label{fig:l1_two_regions}
\end{figure}

We can now prove Theorem~\ref{thm:L1-quadtree}.

\begin{proof}[Proof of Theorem~\ref{thm:L1-quadtree}]
As mentioned, the proof is similar to that for Theorem~\ref{lem:refines}:
we will use the quadtree technique and then add a case distinction 
to handle each region in a more refined manner.

Recall that in the $L_2$-metric, \Ptwo could win voters from at most three blocks, 
which implied that  \Ptwo could only win  $39 \frac{n}{k}$ voters if
\Pone uses the basic quadtree strategy. 
Lemma~\ref{lem:l1_noopposite} tells us that in the $L_1$-metric,
\Ptwo can win voters from at most two blocks. Using a similar analysis
as before, this now implies that \Ptwo can only win  $26 \frac{n}{k}$ voters if
\Pone uses the basic quadtree strategy. 
Next, we show how to improve this bound by incorporating $\eps$-nets.
\medskip

Let $\epsr{k}$ be such that there exists an $\epsr{k}$-net for axis-parallel rectangles
consisting of $k$ points. 
As established~\cite{Aronov2009,dulieu2006,Rachek2020}, we have $\epsr{k} \leq \frac{2}{k+3}$ for $k \leq 5$.
For larger $k$, the current best values are achieved by recursively applying Lemma~4.2 
of Aronov~\etal\cite{Aronov2009}; we show those used in Table~\ref{tab:epsr}.
(Note: if, as conjectured, $\epsr{k} \leq\frac{2}{k+3}$, then it is better to not use the quadtree technique at all.)

\begin{table}[bt]
    \centering
    \begin{tabular}{c|c|c|c|c|c|c|c|c|c|c|c|c|c}
        $k$ & 1 & 2 & 3 & 4 & 5 & 7 & 8 & 10 & 12 & 14 & 16 & 17 & 19 \\ \hline
        $\epsr{k}$ & $1/2$ & $2/5$ & $1/3$ & $2/7$ & $1/4$ & $2/9$ & $1/5$ & $1/6$ & $2/13$ & $1/7$ & $2/15$ & $1/8$ & $2/17$
    \end{tabular}
    \caption{The best known values $\epsr{k}$ such that there exists an $\epsr{k}$-net of size $k$. The value for $k \leq 20$ is only shown if it improves upon $k-1$.}
    \label{tab:epsr}
\end{table}

To make optimal use of the currently known bounds on $\epsr{k}$, we will
enforce that \Ptwo can only win $\frac16 m$ voters from any block.
Let $R \in \R$ be a region that contains $n_R$ voters. Then player \Pone will
place points as follows.
\begin{itemize}
\item
If $n_R \leq \frac16m / \epsr i$ for some $i \leq 19$, \Pone places eight points to 
make $R$ a type-II block and $i$ points according to an $\epsr i$-net, using $8+i$ points in total.
This ensures that \Ptwo can only win $\frac16m$ voters from~$R$.
The following table shows the bounds that we obtain.  \\
    \begin{center}
    \begin{tabular}{l|c|c|c|c|c}
        Number of points ($=8+i$) & 20 & 22 & 24 & 25 & 27 \\ \hline
        min.\ voters covered & $m$ & $1\frac1{12}m$ & $1\frac16m$ & $1\frac14m$ & $1\frac13m$ \\
        max.\ points per $m$ voters & 20 & $20\frac4{13}$ & $20\frac47$ & 20 & $20\frac14$ \\[3mm]
    \end{tabular}
    \end{center}
The table skips some possible configurations, such as the one with $21$ points.
This is because limiting us to the shown cases does not give a worse overall bound but prevents the table from getting overly large.
For any configuration, the minimum number of voters it covers is determined by the maximum number of voters the previous one can handle, e.g.\ $1\frac1{12}m = \frac16m / \epsr{12}$.
Note that we always have $\frac{\mbox{\tiny number of points placed}}{\mbox{\tiny number of voters in $R$}} \leq 20\frac47$.
\item 
If $n_R > \frac16m / \epsr{19} = 1\frac5{12}m$, player~\Pone places 13 points to make 
$R$ a type-I block and places further points as $\epsilon$-nets in each child region of $R$ to ensure at most $\frac{1}{6} m$ voters can be won from it.
Since a child region contains at most $m$ voters and $\epsr{10} \leq 1/6$, we will use $\eps$-nets of size up to ten.
Now let's determine a lower bound on the smallest possible number of voters in~$R$ when \Pone uses $13+i$~points.
If child region $j$ uses $i_j$ points, then it must contain more than $\frac16m / \epsr{i_j-1}$ voters, where we define $\epsr{i_j-1} := \infty$.
Based on this, we define
\[
f(i) := \min\left\{ \sum_{j=1}^4 \frac16m / \epsr{i_j-1}\ \mid \ \sum_{j=1}^4 i_j = i\text{ and }i_j \leq 10\text{ for all }j \right\}.
\]
When \Pone uses $13+i$ points for region~$R$, then $R$ must contain more than $f(i)$ voters.
In other words, $13+i$ points always suffice for up to $f(i+1)$ voters.

The table below shows the bounds this gives, where we can again skip many configurations.
When \Pone places $29$ points, the minimum number of voters covered comes directly from the condition $n_R > 1\frac5{12}m$.
In this case, we also know that $R$ contains at most $f(17)$ voters.
From the equation above, we can deduce that $f(17) = 1\frac12m$, which is achieved by taking $i_1=10$, $i_2=7$ and $i_3=i_4=0$.
Thus, \Pone only places 30 points when $R$ contains more than $1\frac12m$ voters.
Similarly, these 30 points suffice for up to $f(18) = 1\frac7{12}m$ voters (given by $i_1=10, i_2=8, i_3=i_4=0$) and hence when \Pone places 32 points, $R$ contains more than $1\frac7{12}m$ voters.
This continues until we reach $53$, which is $13 + 4 \cdot 10$ and therefore works for any number of voters.
\\
    \begin{center}
    \begin{tabular}{l|c|c|c|c|c|c|c}
        Number of points & 29 & 30 & 32 & 34 & 41 & 51 & 53 \\ \hline
        $i_1, i_2, i_3, i_4$ & & {\tiny 10,7,0,0} & {\tiny 10,8,0,0} & {\tiny 10,10,0,0} & {\tiny 10,10,2,0} & {\tiny 10,10,9,0} & {\tiny  10,10,10,9}\\
        min.\ voters covered & $1\frac5{12}m$ & $1\frac12m$ & $1\frac7{12}m$ & $1\frac23m$ & $2m$ & $2\frac12m$ & $3\frac13m$ \\
        max.\ points per $m$ voters & $20\frac8{17}$ & 20 & $20\frac4{19}$ & $20\frac25$ & $20\frac12$ & $20\frac25$ & $15\frac9{10}$ \\[2mm]
    \end{tabular}
    \end{center}
\end{itemize}
Thus in both cases, the total number of points is $k < 20\frac47 \cdot \frac nm$, meaning $m < 20\frac47 \cdot \frac nk$.
Hence, player~\Ptwo can win at most $\frac13m < 6\frac67 \cdot \frac nk$ voters.
\end{proof}
\section{Conclusion}\label{ch:conclusion}
We studied the discrete one-round Voronoi game where player~\Pone 
can place $k>1$ points and player~\Ptwo can place a single point. We improved the
existing bounds on the number of voters player~\Pone can win. For small $k$
this was done by proving new bounds on $\epsilon$-nets for convex ranges.
For large $k$ we used a quadtree-based approach, which uses the $\eps$-nets
as a subroutine. 
The main open problem is: Can player~\Pone always win at least half
the voters in the $L_2$-metric by placing less than four points?

\bibliographystyle{plainurl}
\bibliography{literature-new}

\end{document}